\documentclass[11pt]{article}
\usepackage{xspace}  
\usepackage{textcomp}
\usepackage{amsmath,amsthm,amsfonts,amssymb}
\usepackage{hyperref}
\usepackage{booktabs}
\usepackage{fullpage}
\usepackage{mathtools}
\usepackage{url}
\usepackage{tcolorbox}
\usepackage{comment}
\usepackage{makecell}
\usepackage[]{algorithm2e, algorithmicx, algpseudocode} 
\usepackage{bbm}
\usepackage{wrapfig}
\usepackage{mathdots}
\RestyleAlgo{boxed}
\usepackage{thmtools} 
\usepackage{thm-restate}
\usepackage[capitalise]{cleveref}
\usepackage[normalem]{ulem} % use \sout
\usepackage{url}
\usepackage{MnSymbol,wasysym} % smiley

\usepackage{cite} % this package doesn't seem to work well when we use ACM style files.

\newcommand{\defn}[1]{\textbf{\emph{\boldmath #1}}}

\newcommand{\poly}{\text{poly}}

\newcommand{\thmref}[1]         {Theorem~\ref{thm:#1}}

\newcommand{\lemref}[1]         {Lemma~\ref{lem:#1}}

\newcommand{\corref}[1]         {Corollary~\ref{cor:#1}}

\renewcommand{\paragraph}[1]{\vspace{0.1in}\noindent{\bf \boldmath #1}} 

\renewcommand{\subparagraph}[1]{\vspace{0.1in}\noindent{\it  #1}} 

\newcommand{\dplus}{d^+}
\newcommand{\dmax}{d^{\max}}
\newcommand{\eccplus}{\ecc^+}
\newcommand{\eccmax}{\ecc^{\max}}

\newcommand{\ecc}{\texttt{ecc}}

\newcommand{\meetmaxdiam}{\ensuremath{ \text{meet}^{\max}\text{-diam}}}
\newcommand{\meetplusdiam}{\ensuremath{\text{meet}^{+}\text{-diam}}}
\newcommand{\meetmaxdist}{\ensuremath{ \text{meet}^{\max}\text{-distance}}}
\newcommand{\meetplusdist}{\ensuremath{\text{meet}^{+}\text{-distance}}}

\newcommand{\meetmaxradius}{\ensuremath{ \text{meet}^{\max}\text{-radius}}}
\newcommand{\meetplusradius}{\ensuremath{\text{meet}^{+}\text{-radius}}}

\newcommand{\kmeetmaxdiam}[1]{\ensuremath{#1 \text{-meet}^{\max}\text{-diam}}}
\newcommand{\kmeetplusdiam}[1]{\ensuremath{#1 \text{-meet}^{+}\text{-diam}}}

\newcommand{\outball}{\ensuremath{\overrightarrow{B}}}
\newcommand{\inball}{\ensuremath{\overleftarrow{B}}}

\newtheorem{theorem}{Theorem}
\newtheorem{lemma}{Lemma} 
\newtheorem{corollary}{Corollary} 
\newtheorem{claim}[lemma]{Claim} 

\newtheorem{definition}{Definition}

\newtheorem{hypothesis}{Hypothesis}

\usepackage{todonotes}

\newif\ifcomments
\commentstrue
\definecolor{yaleblue}{rgb}{0.06, 0.3, 0.57}
\definecolor{ao(english)}{rgb}{0.0, 0.5, 0.0}

\title{When Shall We $k$ Meet Again? Tight Algorithms for Diameter and Radius under the Meet Distance}
\author{Yael Kirkpatrick\thanks{\texttt{yaelkirk@mit.edu}}\\MIT \and John Kuszmaul\thanks{\texttt{jhnkszml@mit.edu}}\\MIT  \and  Merey Temirzinova\thanks{\texttt{merey@mit.edu}}\\MIT \and Virginia Vassilevska Williams\thanks{\texttt{virgi@mit.edu}}\\MIT}
\date{}

\begin{document}

\sloppy

\maketitle

\begin{abstract}
    Finding an optimal meeting point for a collection of agents on a directed graph is a classical problem studied in the context of network analysis, operations research and computational geometry. In this work, we use the two objectives of optimal meeting points examined in the literature to study two notions of meet-distance: $d^{\max}(u,v)$, the minimum over all meeting points $w$ of $\max(d(u,w), d(v,w))$; and $d^+(u,v)$, the minimum over all meeting points of $d(u,w) + d(v,w)$. These values measure the minimum time and minimum total distance required for two agents to meet.

We initiate the fine-grained study of fundamental graph parameters under the two notions of meet-distance, namely the diameter, radius and eccentricities. 

For general directed graphs, we give an $\tilde{O}(m\sqrt{n})$ time algorithm for computing a 2-approximation to both notions of meet-diameter and show that this result is optimal under SETH. In contrast, we show that such a result is unattainable for the  meet-radius as any finite approximation requires quadratic time under the Hitting Set Conjecture.

For directed acyclic graphs, we obtain stronger results. We compute the meet$^{\max}$-diameter exactly in linear time and give a linear-time $2$-approximation for the meet$^{+}$-diameter. We complement the latter with a quadratic-time lower bound for any $(3/2-\varepsilon)$-approximation under SETH, yielding a separation between the two meet-distance objectives.

Finally, we study the generalized meet-distance of $k$-tuples of vertices. For every positive integer $\ell$,
we reduce the problem of $\ell$-approximating the $k$-point meet-diameter to computing an exact meet-diameter on smaller tuples, obtaining an $\ell$-approximation in
time
\[
O\!\left(
mn+ \ell \left\lceil k^{1/\ell}\right\rceil
n^{\left\lceil k^{1/\ell}\right\rceil+1}
\right).
\]

\end{abstract}

\section{Introduction}
The problem of finding an optimal meeting point for two or more agents in a graph is one of the oldest questions in network analysis, raised by Hakimi in the 1960s \cite{hakimi1964}. The goal of this original work was to find an optimal location for the placement of a switching center in a communication network or for a hospital in a city.

These two goals motivate two slightly different problem definitions. 
The hospital placement aims to minimize the farthest a single person must travel to receive care. If we denote by $d(u,v)$ the distance from point $u$ to point $v$, our cost function in this case is
\[
\min_{v}\max_{u} d(u,v).
\]

On the other hand, the switching center placement aims to minimize the total length of wires used to connect all members of the network to the station. In this case, our cost function becomes
\[
\min_{v}\sum_{u} d(u,v).
\]

The optimal location for the hospital is known as the \emph{absolute center}, while the optimal location for the switching center is known as the \emph{absolute median}. In the context of Euclidean distances, the $(\min, +)$ switching station placement objective is known as the `Weber Problem' and has been studied extensively in the field of operations research and  computational geometry \cite{Weiszfeld1937, kuhn1962, cooper1968, chen1984, ostresh1977}. The $(\min, \max)$ hospital placement objective has also seen a long line of work \cite{Welzl1991,Megiddo1983}. 

More recently, optimal meeting-point problems (OMP) have been studied from the perspective of spatial databases aiming to answer queries of the form `find an optimal meeting point for a given set of agents' \cite{processingMeetPoints15, lee2021efficient, optimalMeetRoads2023, dynamicOMP2018, otaki2022planning}. The OMP literature also focuses on precisely the two objectives mentioned above: namely the min-max objective, which minimizes the farthest any single agent has to travel to the meeting point (hospital objective); and the min-sum objective, which minimizes the total distance traveled by all agents (switching center objective).

In this work we  revisit this classical optimization problem from a different perspective. Rather than asking for the optimal meeting location of some queried groups, we use optimal meeting as a notion of distance between vertices. Formally, given a directed, weighted graph $G = (V,E)$ we define two notions of meet-distance between a pair of points

\[
\dmax(u,v) \coloneqq \min_{w\in V} \max(d(u,w), d(v,w)),~~~~~~ \dplus(u,v) \coloneqq \min_{w\in V} d(u,w) + d(v,w).
\]

We call these \meetmaxdist{} and \meetplusdist. These are natural notions for directed graphs. While in undirected graphs, the $\dplus$ distance essentially collapses to the standard distance, in directed graphs two vertices may be unable to reach one another while still being able to reach a common third vertex. By considering $u$ and $v$ as potential meeting points themselves (using the standard definition $d(u,u) = 0$) we have that $d^{\max/+}(u,v) \leq \min (d(u,v), d(v,u))$ or the min-distance of $(u,v)$ \cite{fixedparam16}.

We note that these two distance notions are within a factor 2 of each other, as considering the optimal meeting point for one notion as a potential meeting point for the other gives us:
\[
\dmax(u,v) \leq \dplus(u,v) \leq 2 \dmax(u,v).
\]

Under these notions of distance, we explore the problem of computing fundamental graph parameters such as diameter, radius and eccentricities. The diameter, or the largest distance in the graph, is a well-studied problem in fine-grained complexity as it is a natural parameter measuring how fast information can spread in a network, or in our case, how quickly agents can meet. Computing the optimal center under the meet-distance can determine the optimal location to place a distribution center, allowing everyone to meet with someone coming from the distribution center to receive their service (care, medication, etc.).

The diameter, as well as all node-eccentricities (the largest distance from each node) and radius (the smallest node eccentricity) can be computed exactly by measuring all distances in a graph, but as this can take a long time \cite{fixedparam16}, a long line of work has gone into approximating these parameters quickly or proving hardness of such approximations. This includes undirected graphs \cite{sprasediamradius13, cgr, aingworthdiam99, undirdiamlb, towardsapprox21, betterapproxdiam14}, directed graphs \cite{bonnet21, lidiamapprox21, dirdiamlb21, fixedparam16}, as well as non-standard notions of distance such as the min-distance and max-distance \cite{fixedparam16, mindiam2022, mindiamdag21, mindistance2019, bergermindiam2023, improvedlinearmindiam2026}, roundtrip-distance \cite{roundtrip99}, and multimode-distance \cite{multimode, improvedlinearmindiam2026}.

In this work we study the problem of approximating these central graph parameters under both notions of meet-distance defined above. We consider the problem in general directed graphs as well as directed acyclic graphs (DAGs).

Finally, as in the literature, we consider the meet-distance of a group ($> 2$) of nodes. We generalize the above definitions to any $k$-tuple.

\begin{align*}
    \dmax(v_1, \ldots, v_k) &\coloneqq \min_{w\in V}\max_{1\leq i \leq k} d(v_i, w),\\
    \dplus(v_1, \ldots, v_k) &\coloneqq \min_{w\in V}\sum_{i=1}^k d(v_i, w).
\end{align*}

\subsection{Our Results}
For the 2-point case we are able to paint an almost complete picture, summarized in \cref{tab:2pointresults}. First we show that a single node eccentricity can be computed in near-linear time, as with standard graph distance, and thus all distances can be computed in time $\tilde{O}(mn)$. We then show that any finite approximation to the meet-radius, under either notion of meet-distance (and thus also any approximation to all-node eccentricities) requires quadratic time under the Hitting Set Conjecture, even when restricting the problem to DAGs.

For both \meetmaxdiam{} and \meetplusdiam{}, on the other hand, we are able to construct a 2-approximation algorithm running in $\tilde{O}(m\sqrt{n})$ time. We prove a matching lower bound, showing that any $(2-\varepsilon)$ approximation requires quadratic time under the Strong Exponential Time Hypothesis (SETH). Finally, in the case of DAGs, we show improved algorithms for both notions of meet-diameter. In fact, we are able to compute the  \meetmaxdiam{} exactly in linear time. We show that this is not possible for the \meetplusdiam{} as any $(3/2-\varepsilon)$ approximation requires quadratic time under SETH, giving a  separation between the two notions of meet-distance. We supplement this lower bound with a 2-approximation to the problem running in linear time, faster than its counterpart in general graphs. 

% Please add the following required packages to your document preamble:
% \usepackage{booktabs}
\begin{table}[]
\begin{tabular}{@{}llllll@{}}
\toprule
\textbf{Problem}  & \textbf{DAG}                 & \textbf{Approx} & \textbf{Runtime}                & \textbf{Lower Bound}   & \textbf{Reference}                                                                  \\ \midrule
$\meetmaxdiam$  & no          & 2      & $\tilde{O}(m\sqrt{n})$ & $(2-\varepsilon, O(1))$   & \cref{thm:2approx}, \cref{thm:2lb}       \\
$\meetplusdiam$ & no          & 2      & $\tilde{O}(m\sqrt{n})$ & $(2-\varepsilon, O(1))$   & \cref{thm:2approx}, \cref{thm:2lb}       \\
$\meetmaxdiam$ & yes   & exact  & $\tilde{O}(m)$         &               & \cref{clm:dagmaxdiag}                                     \\
$\meetplusdiam$ & yes   & 2      & $\tilde{O}(m)$         & $(3/2-\varepsilon, O(1))$ & \cref{clm:dagplusdiam}, \cref{thm:daglb} \\
$\meetmaxradius$ & yes          &       &  & any   &  \cref{thm:lbradius}  \\
$\meetplusradius$ & yes         &       &  & any   &  \cref{thm:lbradius}  \\
 \bottomrule
\end{tabular}
\caption{Result Summary for 2-point Meet-Distance. All lower bounds rule out $O(m^{2-\delta})$-time algorithms for any constant $\delta>0$.} \label{tab:2pointresults}
\end{table}

For our generalized $k$-point meet-distance we can no longer compute all distances in time $\tilde{O}(mn)$, as there are $n^k$ possible tuples to consider. For this problem, we are able to show a reduction from $\ell$-approximate $k$-point meet-diameter to a weighted notion of exact $k'$-point meet-diameter for values $k' \leq k^{1/\ell}$. This allows us to get an $\ell$-approximation to both types of $k$-meet-diam in time $O(mn + \ell \left\lceil k^{1/\ell}\right\rceil n^{\left\lceil k^{1/\ell}\right\rceil + 1})$, for any positive integer $\ell$.

\subsection{Technical Overview}

We highlight a few key techniques in our results, demonstrating where this new distance metric retains similar structure to ordinary distance and where its definition as a global optimization problem makes it behave very differently. 

The first notable difference between meet-distance and ordinary distance is that meet-distance is not a metric as it does not satisfy the triangle inequality. This renders many traditional graph parameter approximation techniques inapplicable. However, we can show  a more restricted, one-side triangle inequality by relaxing one of the distances to an ordinary distance:
\[
d(a,x) + d^{\max / +}(x,b) \geq d^{\max / +}(a,b).
\]

Conceptually, this inequality tells us that $a$ and $b$ can meet by first $a$ going to any vertex $x$, and then meeting at the optimal meeting point for $x$ and $b$. This interaction between ordinary and meet-distance implies that in a diameter $D$ graph, any vertex falling within ordinary distance $D/2$ from a diameter endpoint will have meet-distance $\geq D/2$ from the other diameter endpoint, giving us a 2-approximation to the diameter, i.e. if $d^{\max/+}(a,b) = D$ and $d(a,x) \leq D/2$ then $d^{\max/+}(x,b) \geq D/2$.

Using this property we introduce a shrinking argument - if a vertex has a large incoming $D/2$ neighborhood we compute its meet-eccentricity. If this meet-eccentricity is not $\geq D/2$ we know that no diameter endpoint lies in this incoming neighborhood and thus we can remove the neighborhood from the set of potential diameter endpoints. We do this until all nodes have small incoming $D/2$ neighborhoods. Now we note that any two points $u,v$ with meet-distance $\leq D/2$ must meet at a point $x$ such that both $u$ and $v$ fall in the incoming $D/2$ neighborhood of $x$. Since this incoming neighborhood is small for every $x$ we can use this argument, together with some additional delicate details (see \cref{sec:diamapprox}), to find a point $u$ that cannot meet all other points at meet-distance $\leq D/2$, even when considering all potential optimal meeting points. This point $u$ therefore has meet-eccentricity $\geq D/2$ are we are finished.

When considering $k$-tuples instead of pairs, we exploit this compositional nature of meeting points in a different way. We break the problem of points $v_1, \ldots, v_k$ meeting into phases - first each smaller tuple of $k'$ points meets: $v_1, \ldots, v_{k'}$ meet at $u_1$, $v_{k'+1}, \ldots, v_{2k'}$ meet at $u_2$ and so on. This gives us $k/k'$ meeting points $u_1, \ldots, u_{k/k'}$ and now we find an optimal meeting point for them. Thus, we have reduced the problem of finding an (approximately) optimal meet point for $k$ vertices into $k/k'+1$ problems of finding an optimal meeting point for a smaller tuple of $k'$ or $k/k'$ points. This is the core idea of our $2$-approximation to $k$-point meet-diameter. To extend this idea to an $\ell$-approximation, we perform a similar argument over a weighted tree of depth at most $\ell$.

In DAGs, on the other hand, we note a strong property that is not present in regular graphs, revealing a sharp distinction between the two objectives. Every DAG has a sink, and for a sink $u$ and any other vertex $v$ we note that $d^{\max/+}(v,u) = d(v,u)$ as the only potential meeting point for this pair is $u$ itself. This means that the largest distance $d(w,u)$ from $u$ is a true meet-distance in the graph and so at most the diameter. On the other hand, by considering $u$ as a potential meeting point for the \meetmaxdiam{}-endpoints, we also have that this distance is a lower bound on the diameter. Meaning this easy-to-compute distance is in fact equal to the \meetmaxdiam{}. This same approach gives a 2-approximation when considering the \meetplusdist{} objective.

Finally, our lower bounds exploit the meeting-point structure in yet another way. We consider the standard problems used to reduce to graph diameter and radius problems - the orthogonal vectors and hitting set problems, which are hard under standard fine-grained assumptions. In both of these problems we have a set of vectors with some question about pairwise orthogonality, meaning that we are asking whether a pair of vectors share a coordinate at which they are both equal to 1. We encode this by forcing the coordinates to become meeting points - when two vectors are non-orthogonal they can simply meet at this shared coordinate. This idea underlies our various lower bound constructions.

\section{Preliminaries}
Let $G=(V, E)$ be a directed weighted graph with $n = |V|$ nodes and $m = |E|$ edges. We assume in this paper that all edge weights are nonnegative integers bounded by $\poly (n)$. Denote by $d(u,v)$ the length of the shortest path in $G$ from $u$ to $v$ and for a set $S$ define $d(u,S) = \min_{s\in S}d(u,s)$. We say $d(u,v) = \infty$ if no such path exists. If $G$ contains no cycles we call it a directed acyclic graph, or DAG.

Define the outgoing ball of radius $r$ around a point $x$ as $\outball(x,r) = \{y\in V : d(x,y) \leq  r\}$ and the incoming ball as $\inball(x,r) = \{y\in V : d(y,x) \leq  r\}$.

Given a parameter $T$ we wish to compute, we say an algorithm returns an $(\alpha, \beta)$-approximation to $T$ if it outputs a value $\tilde{T}$ satisfying $T \leq \tilde{T}\leq \alpha\cdot T + \beta$, or equivalently $\frac{T - \beta}{\alpha}\leq \tilde{T}'= \frac{\tilde{T}-\beta}{\alpha}\leq T$.

Define the following two notions of meet-distance:

\begin{definition}[\meetmaxdist]
    $\dmax(u,v)\coloneqq \min_{w\in V} \max(d(u,w), d(v,w))$.
\end{definition}

\begin{definition}[\meetplusdist]
    $\dplus(u,v)\coloneqq \min_{w\in V} d(u,w) +d(v,w)$.
\end{definition}

For any vertex $v\in V$ we define its eccentricity to be the largest distance $d(v,u)$ for some $u\in V$. The largest eccentricity in the graph (or the largest distance in the graph) is known as the \emph{diameter}, while the smallest eccentricity is known as the \emph{radius}. We extend these definitions to the meet-distance and define meet-distance eccentricities, denoted by $\eccmax(v), \eccplus(v)$, and analogous notions of diameter and radius, denoted by \meetmaxdiam, \meetplusdiam, \meetmaxradius{} and \meetplusradius.

Next, we generalize these definitions of meet-distance to $k$-tuples of points.

\begin{definition}[$k$-\meetmaxdist]
    $\dmax(v_1, \ldots, v_k) \coloneqq \min_{w\in V}\max_{1\leq i \leq k} d(v_i, w)$.
\end{definition}

\begin{definition}[$k$-\meetplusdist]
    $\dplus(v_1, \ldots, v_k) \coloneqq \min_{w\in V}\sum_{i=1}^k d(v_i, w)$.
\end{definition}

We define the \kmeetmaxdiam{k} and  \kmeetplusdiam{k} accordingly.

In our algorithm we make use of the following framework of Cohen \cite{Cohen97} for estimating the size of the neighborhoods of all vertices in the graph.

\begin{lemma}\label{lem:cohenframework}[Neighborhood Estimation\cite{Cohen97}]
    Given a directed $n$-node $m$-edge graph $G$, parameters $r, \varepsilon$ and a subset $W\subseteq V$, there exists an algorithm running in time $\tilde{O}(m\varepsilon^{-2})$ that computes with high probability an estimate $\hat{n}(v)$ for each vertex $v\in V$ satisfying 
    \[
    ||\inball(v, r)\cap W| - \hat{n}(v)| \leq \varepsilon |\inball(v, r)\cap W|.
    \]
\end{lemma}

This lemma is obtained from Theorem 5.1 of \cite{Cohen97} by switching the direction of the edges (taking the graph with edges in reverse direction) and limiting the random ranking only to vertices of $W$.

\subsection{A Simple Triangle Inequality}
\begin{lemma}[Meet Triangle Inequality]
\label{lem:triangle}
For any $a, x, b \in V$,
\[d(a, x) + \dmax(x, b) \geq \dmax(a, b)\]
and
\[d(a, x) + \dplus(x, b) \geq \dplus(a, b).\]
\end{lemma}
\begin{proof}

    Assume $d(a,x) < \infty$ as otherwise the claim is trivial. Consider a vertex $y$ that is an optimal meeting point for $x$ and $b$, i.e., 
    \begin{align*}
    \dmax(x, b) &= \max(d(x, y), d(b, y)) \\
    &\geq \max(d(a, y) - d(a, x), d(b, y)) \\
    &\geq \max(d(a, y), d(b, y)) - d(a, x) \\
    &\geq \dmax(a, b) - d(a, x),
    \end{align*}
    where the second line follows via the standard triangle inequality. Similarly,

    \begin{align*}
    \dplus(x, b) &= d(x, y) + d(b, y) \\
    &\geq d(a, y) - d(a, x) + d(b, y) \\
    &\geq \dplus(a, b) - d(a, x).
    \end{align*}

\end{proof}
 We observe a more intuitive proof of the lemma. Under both notions of meet-distance, the meet-distance for two agents starting at $a$ and $b$ is at most the time it takes for the first agent to reach $x$ from $a$ plus the meet-distance between two agents starting at $x$ and $b$.

 \subsection{Computing Eccentricity}

First, we consider how to compute the \meetmaxdist{} \defn{eccentricity}, $\eccmax(x) := \max_y \dmax(x, y)$. Similarly, $\ecc^{+}(x) := \max_y \dplus(x, y)$.

\begin{lemma} \label{lem:maxecc}
    We can compute $\eccmax(x)$ in time $\tilde{O}(m)$.
\end{lemma}

\begin{proof}
Given a value $r$ we can perform the following algorithm. We first compute $W = \outball(x, r)$ via a forwards Dijkstra search from $x$. We then compute $U = \cup_{w \in W} \inball(w, r)$ via a reverse\footnote{A Dijkstra search on the graph with reversed edges.} Dijkstra search on the set $W$. This gives us $U = \{v \in V \mid \dmax(x, v) \leq r\}$. If $U = V$, then we know $\eccmax(x) \leq r$. Otherwise, $\eccmax(x) > r$. This has runtime $\tilde{O}(m)$. Thus, by performing a binary search over $r$ we can find the value of $\eccmax(x)$ in $\tilde{O}(m)$ time.
\end{proof}

\begin{lemma}\label{lem:plusecc}
    We can compute $\eccplus(x)$ in time $\tilde{O}(m)$.
\end{lemma}

\begin{proof}
    We first compute a Dijkstra search out of $x$ to find $d(x, v)$ for all vertices $v$. We then create a new vertex $t$, in a modified graph $G'$, and for every vertex $v$, we add an edge $(v, t)$ of weight $d(x, v)$. We then perform a Dijkstra search into $t$. We claim that the vertex $v$ with largest distance $d(v, t)$ will satisfy $\eccplus(x) = d(v, t)$. To prove this claim, we first observe that $\dplus(x, v) = d(v, t)$. This follows since, if the optimal meeting point of $x$ and $v$ in $G$ is at some vertex $u$, then there is a path from $v$ to $u$ and then $u$ to $t$ in $G'$ with the same cost as $\dplus(x, v)$. Any shorter path utilizing an edge $(u', t)$ would imply that meeting at $u'$ has lower cost than meeting at $u$ in $G$, and thus $u$ is not an optimal meeting point for $v$ and $x$, a contradiction. Thus maximizing over all $d(v, t)$, we find the \meetplusdist{} eccentricity of $x$. This requires only two applications of Dijkstra's algorithm, which take $\tilde{O}(m)$ time.
\end{proof}

We note that the above algorithms show us that in near-linear time we can compute \emph{all \meetplusdist s} out of a vertex. On the other hand, we can only compute \meetmaxdist s out of a vertex $x$ with respect to a threshold, i.e. which vertices are within \meetmaxdist{} $r$ of $x$ and which are \meetmaxdist{} $>r$ away. In the $\dplus$ algorithm we are able to compute the optimal meeting point for each pair of vertices by keeping track of the final edge of the paths coming into our searched vertex in the augmented graph.

We conclude that for both notions of meet-distance we can compute all distances in the graph (and thus its exact diameter, radius and eccentricities) in time $\tilde{O}(mn)$, matching the runtime of the standard distance metric.  

\begin{corollary}
    We can compute all-node meet$^{\max}$- and meet$^+$-eccentricities (giving the exact diameter and radius) of a directed $n$-node $m$-edge graph in time $\tilde{O}(mn)$. In the same time we can also compute the \meetplusdist{} between every pair of points (APSP) as well as their optimal meeting point.
\end{corollary}

\subsection{Fine-Grained Hypotheses}

We prove our diameter lower bounds through a reduction from the orthogonal vectors (OV) problem, defined as follows. Given a set $A$ of $n$ $d$-dimensional boolean vectors, the OV-problem asks if there exists an orthogonal pair, $a,b\in A$ such that $a\cdot b = 0$. Williams  \cite{williams2005} showed that, in the word-RAM model of computation with $O(\log n)$ bit words, when $d=\omega(\log n)$, solving the OV problem requires $n^{2-o(1)}$ time under the Strong Exponential Time Hypothesis (SETH) \cite{impagliazzopaturi2001, impagliazzo2001problems}.

\begin{hypothesis}[Orthogonal Vectors Hypothesis, implied by SETH]
    No algorithm running in time $O(n^{2-\delta})$ can solve the orthogonal vectors problem on $n$ vectors of dimension $\omega(\log n)$ for any $\delta>0$, in the word-RAM model of computation with $O(\log n)$ bit words.
\end{hypothesis}

We cast the OV-problem as a problem on a graph with $n$ vertices representing the vectors of $A$ and $d$ vertices representing the coordinates $[d]$. As in prior work, we add an edge between $a\in A, c\in [d]$ when $a[c]=1$. 

We can assume w.l.o.g that every vector $a\in A$ has some coordinate $i$ for which $a[i] = 1$ and every coordinate $i$ has some $a\in A$ for which $a[i]=1$ as otherwise we can determine the answer to the problem in the first case, or ignore this coordinate and remove it from the construction in the latter case.

For radius lower bounds we often require different quantifiers in our base problem. We prove our radius lower bounds through a reduction from the hitting set (HS) problem, defined as follows. Given two sets $A,B$ of $n$ $d$-dimensional vectors, the HS-problem asks if there exists a vector $a\in A$ that \emph{hits}  every vector $b\in B$, i.e. $\exists a\in A $ such that $\forall b\in B~~a\cdot b \neq 0$. This problem was introduced by Abboud, Vassilevska W. and Wang and conjectured to take $n^{2-o(1)}$ time \cite{fixedparam16}.

\begin{hypothesis}[Hitting Set Conjecture\cite{fixedparam16}]
    No algorithm running in time $O(n^{2-\delta})$ can solve the hitting set problem on $n$ vectors of dimension $\omega(\log n)$ for any $\delta>0$, in the word-RAM model of computation with $O(\log n)$ bit words. 
\end{hypothesis}

We represent the HS-problem in an analogous way to our representation of the OV-problem.  We can assume w.l.o.g that every vector  $b\in B$ has some coordinate $i$ for which $b[i] = 1$ and every coordinate $i$ has some $a\in A$ for which $a[i]=1$ as otherwise we can again determine the answer to the problem in the first case, or ignore this coordinate and remove it from the construction in the latter case.

\section{Diameter and Radius Approximation}\label{sec:diamapprox}
In this section we show a tight approximation to the meet-diameter. First we show that in $\tilde{O}(m\sqrt{n})$ time we can compute a 2-approximation to both notions of meet-diameter. In the lower bounds section we show that this approximation factor is in fact optimal and no subquadratic time algorithm can achieve a $(2-\varepsilon)$-approximation.

We improve this approximation for DAGs. Here we have different results for the two notions of meet-distance; in linear time, we are able to compute the exact \meetmaxdiam{} of a DAG. Under our lower bounds, this cannot be done for \meetplusdiam. We instead achieve a $2$-approximation of the \meetplusdiam{} in linear time.

For radius we show an entirely different picture: we prove a lower bound for both notions of meet-radius showing that \emph{any finite approximation} to the meet-radius requires quadratic time, even in DAGs.

\subsection{Diameter Approximation Algorithms}

\begin{theorem}\label{thm:2approx}
There is a $2$-approximation for \meetmaxdiam{} and \meetplusdiam{}, running in time $\tilde{O}(m\sqrt{n})$.
\end{theorem}

We prove \thmref{2approx} via a sequence of lemmas. We prove the theorem for \meetmaxdist{} and note that the same proof holds for \meetplusdist. This is because the only properties of the \meetmaxdist{} we use are \cref{lem:triangle} and \cref{lem:maxecc}, which hold true for \meetplusdist{} as well (as Lemmas \ref{lem:triangle} and  \ref{lem:plusecc} respectively).

Throughout, we assume we are given a `guess' $D$ for the true value of the meet-diameter. If the meet-diameter is at least $D$ then our algorithm outputs a pair of points at meet-distance $\geq D/2$. Thus, by performing a binary search over $D$ we obtain a $2$-approximation to the meet-diameter using $O(\log n)$ guesses\footnote{Using a standard technique for graph parameter approximation - we pick a value $D$ and use this algorithm to find out if the diameter is $\geq D/2$ or $<D$. As the possible values of $D$ in a graph with edge weights bounded by $W$ are between $1$ and $n\cdot W = \poly(n)$, this binary search takes logarithmically many queries.}, incurring only a logarithmic overhead.

Given $D$, assume it is in fact a lower bound to the meet-diameter and fix a pair of diameter endpoints $a,b$, i.e. a pair such that $\dmax(a,b) = \meetmaxdiam\geq D$.

\begin{lemma}[Sufficient condition]\label{lem:sufficient}
If $x \in V$ satisfies $d(a, x) \leq D/2$, then $\eccmax(x) \geq D/2$.
\end{lemma}
\begin{proof}
By \lemref{triangle}, $\dmax(x, b) \geq \dmax(a, b) - d(a, x) \geq D - D/2 = D/2$.
\end{proof}

\begin{lemma}[Shrinking lemma]\label{lem:shrink}
There is an algorithm running in time $\tilde{O}(m\sqrt{n})$ that either:
\begin{enumerate}
    \item[(i)] outputs a vertex $v$ with $\eccmax(v) \geq D/2$, or
    \item[(ii)] outputs a set $W \subseteq V$ with $a, b \in W$ satisfying one of:
    \begin{itemize}
        \item[(a)] $|W| = O(\sqrt{n})$, or
        \item[(b)] for all $v \in V$, $|\inball(v, D/2) \cap W| < |W|/\sqrt{n}$.
    \end{itemize}
\end{enumerate}
\end{lemma}
\begin{proof}
Initialize $W = V$. At each step, if $W$ already satisfies (a) or (b), we terminate. Otherwise, there exists a vertex $v$ such that $|\inball(v, D/2) \cap W| \geq |W|/\sqrt{n}$. Using \cref{lem:cohenframework} with $\varepsilon=1/3, r = D/2$, we can find a $v$ with $|\inball(v, D/2) \cap W| \geq |W|/2\sqrt{n}$  in $\tilde{O}(m)$ time.

We compute $\eccmax(v)$ in $\tilde{O}(m)$ time using \cref{lem:maxecc}. If $\eccmax(v) \geq D/2$, we output $v$ and terminate. Otherwise, \lemref{sufficient} implies $d(a, v) > D/2$ and $d(b, v) > D/2$, so $a,b \notin \inball(v, D/2)$. We set $W \leftarrow W \setminus \inball(v, D/2)$, preserving the invariant $a, b \in W$, and reducing $|W|$ by a factor of at least $(1 - 1/2\sqrt{n})$.

After $O(\sqrt{n} \log n)$ iterations, either (a) is satisfied (since $n \cdot (1 - 1/2\sqrt{n})^{2\sqrt{n}\log n} < 1$) or (b) was reached earlier. Each iteration takes $\tilde{O}(m)$ time, for a total of $\tilde{O}(m\sqrt{n})$.
\end{proof}

\begin{proof}[Proof of \thmref{2approx}]
Sample a set $S \subseteq V$ of size $\Theta(\sqrt{n} \log n)$ uniformly at random. With high probability, $S$ intersects every ball $\outball(v, D/2)$ of size at least $\sqrt{n}$. Compute the \meetmaxdist -eccentricity of every vertex in $S$ in time $\tilde{O}(m\sqrt{n})$.

\paragraph{Case 1: $|\outball(a, D/2)| \geq \sqrt{n}$.}
Then $S$ hits $\outball(a, D/2)$, so some $s \in S$ satisfies $d(a, s) \leq D/2$. By \lemref{sufficient}, $\eccmax(s) \geq D/2$, and this value was already computed.

\paragraph{Case 2: $|\outball(a, D/2)| < \sqrt{n}$.}
We invoke \lemref{shrink}. If it returns a vertex $v$ with $\eccmax(v) \geq D/2$, we are done. Otherwise, it returns a set $W$ with $a, b \in W$. We handle the two sub-cases.

\subparagraph{Case 2a: $|W| = O(\sqrt{n})$.}
Since $a \in W$, computing $\eccmax(v)$ for all $v \in W$ recovers $\eccmax(a) \geq D$. This takes $\tilde{O}(m\sqrt{n})$ time.

\subparagraph{Case 2b: $W$ satisfies condition (b).}
Let $w = \arg\max_{u \in W} d(u, S)$. We claim $\eccmax(w) > D/2$, so that computing $\eccmax(w)$ yields the desired approximation.

First, $d(w, S) > D/2$: we can assume that $S$ does not hit $\outball(a, D/2)$, or this case would have been handled as in Case~1, so $d(a, S) > D/2$. Since $a \in W$ and $w$ maximizes $d(\cdot, S)$ over $W$, the claim follows.

Second, $|\outball(w, D/2)| < \sqrt{n}$: if $|\outball(w, D/2)| \geq \sqrt{n}$, then $S$ would hit this ball, giving $d(w, S) \leq D/2$, a contradiction.

Finally, the set of vertices $u$ with $\dmax(w, u) \leq D/2$ is contained in
\[
\bigcup_{z \in \outball(w, D/2)} \inball(z, D/2),
\]
since $\dmax(w, u) \leq D/2$ requires a meeting point $z$ with $d(w, z) \leq D/2$ and $d(u, z) \leq D/2$. Restricting to $W$, this union covers at most $|\outball(w, D/2)| \cdot |W|/\sqrt{n} < \sqrt{n} \cdot |W|/\sqrt{n} = |W|$ vertices of $W$. Hence some vertex of $W$ has meet-distance greater than $D/2$ from $w$, so $\eccmax(w) > D/2$.

\medskip

In all cases, we find a vertex with eccentricity at least $D/2$ and output a pair of vertices with meet-distance $\geq D/2$. The total running time is $\tilde{O}(m\sqrt{n})$.
\end{proof}

\subsubsection{Improved Approximations for DAGs}

We now consider computing or approximating the meet-diameter in directed acyclic graphs (DAGs). In fact, the only property of a DAG that we use is the existence of a sink, so the following claims are true for any directed graph that has a sink vertex.

\begin{claim}\label{clm:dagmaxdiag}
    There is an algorithm running in linear time that computes the exact \meetmaxdiam{} of a DAG.
\end{claim}

\begin{proof}
    Let $u$ be a vertex with no outgoing edges. Run Dijkstra's into $u$ to compute and output $\tilde{D} = \max_{v\in V} d(v,u)$. We now argue that $\tilde{D} = \meetmaxdiam$.

    First, since $u$ is a sink we have that $d(u,w) = \infty$ for any $w\neq u$. Thus for any $v\in V$ 
    \[
    \dmax(v,u) = \min_{w}\max\{d(v,w), d(u,w)\} = d(v,u).
    \]
    Therefore, there exists $v$ such that $\tilde{D} = d(v,u) = \dmax(u,v)$ and so $\tilde{D}\leq \meetmaxdiam$.

    On the other hand, let $a,b$ be a pair of diameter endpoints, $\dmax(a,b) = \meetmaxdiam$. Then,
    \[
    \dmax(a,b) = \min_{w}\max(d(a,w), d(b,w))\leq \max(d(a,u), d(b,u)) \leq d(v,u) = \tilde{D},
    \]
    as $\tilde{D}$ maximizes the distance to $u$. We conclude that $\tilde{D} = \meetmaxdiam$ and can be computed in linear time.
\end{proof}

\begin{claim}\label{clm:dagplusdiam}
    There is an algorithm running in linear time that computes a 2-approximation to the \meetplusdiam{} of a DAG.
\end{claim}

\begin{proof}
    Let $u$ be a vertex with no outgoing edges. Run Dijkstra's algorithm into $u$ to compute and output $\tilde{D} = \max_{v_1\neq v_2} d(v_1,u) +d(v_2, u)$, i.e. the sum of the two largest distances to $u$. Denote by $D^+$ the true \meetplusdiam, we argue that $D^+\leq \tilde{D}\leq 2\cdot D^+$.

    Again, since $u$ is a sink we have that $d(u,w) = \infty$ for any $w\neq u$ and so for any $v\in V~~\dplus(v,u) = d(v,u)$.
    Therefore, there exist $v_1, v_2$ such that \[
    \tilde{D} = d(v_1,u) + d(v_2, u) = \dplus(v_1,u) + \dplus(v_2, u) \leq 2D^+.
    \] 
    
    On the other hand, let $a,b$ be a pair of diameter endpoints, $\dplus(a,b) = D^+$. Then,
    \[
    D^+ = \dplus(a,b) = \min_{w}d(a,w)+ d(b,w)\leq d(a,u)+ d(b,u) \leq \tilde{D}.
    \]
    as $\tilde{D}$ maximizes the sum of such distances to $u$. We conclude that $D^+\leq \tilde{D}\leq 2\cdot D^+$ and so $\tilde{D}/2$ is a $(2,0)$-approximation to the diameter and can be computed in linear time.
\end{proof}

\subsection{Lower Bounds}

In this section we show that a 2-approximation for either variant of meet-diameter is optimal, as any improvement over this approximation factor requires quadratic running time under SETH. We show that in DAGs, approximating the \meetplusdiam{} beyond a $3/2$-factor requires quadratic time, implying no linear-time exact algorithm exists under SETH.

For the radius problem we show that any finite approximation is hard, for either notion of meet-distance, even in the restricted case of DAGs.

In this section, we assume the word-RAM model of computation with $O(\log n)$ bit words. All lower bound theorems are for this model, and we will omit including it in the statements of the theorems.

\subsubsection{Diameter Lower Bounds}

\begin{theorem} \label{thm:2lb}
    Under SETH, computing a $(2-\varepsilon, O(1))$-approximation to the \meetmaxdiam{} or \meetplusdiam{} requires $m^{2-o(1)}$ time.
\end{theorem}

\begin{proof}
    \paragraph{Multiplicative lower bound for \meetmaxdiam}: We begin with a simple lower bound for a pure $(2-\varepsilon)$-approximation to the \meetmaxdiam. Given an instance $A$ of the OV-problem construct a graph $G$ consisting of the set $A$, where each vector is represented by a node, and a set $U$ of $d$ vertices $u_1, \ldots, u_d$.  For every $a\in A$ and $i\in [d]$ we construct an edge from the node $a$ to $u_i$ if the vector $a$ has a 1 at the $i$'th coordinate. Lastly, we add all edges between pairs of points in $U$.

    In this construction, if all pairs of vectors in $A$ are non-orthogonal (the NO-case), the graph $G$ has $\meetmaxdiam = 1$. Indeed, for any pair of vertices $a,b\in A$ the vectors $a,b$ are non-orthogonal so there exists a coordinate $i$ for which $a[i]=b[i]=1$. Thus the edges $(a,u_i)$ and $(b,u_i)$ exist and so $\dmax(a,b) = 1$. For any pair of points     $u,u'\in U$ there exists an edge between them and so $\dmax(u,u')=1$, and for any pair $a\in A, u\in U$ there exists some $u'\in U$ for which the edge $(a,u')$ exists (since it is not the zero-vector). Since the edge $(u,u')$ also exists (unless $u=u'$) we conclude $\dmax(a,u) = 1$.

    On the other hand, if there exists an orthogonal pair $a,b\in A$ (the YES-case), then there exists no $u_i\in U$ for which $d(a,u_i) = d(b,u_i) = 1$ as that would imply that $a,b$ both have $1$ at coordinate $i$ (and thus are not orthogonal). Therefore, $\dmax(a,b) > 1$ so $\dmax(a,b)\geq 2$ and $\meetmaxdiam\geq 2$. 

    Therefore, any $(2-\varepsilon)$-approximation algorithm to \meetmaxdiam{} would solve OV and thus requires $n^{2-o(1)} = m^{2-o(1)}$ time under SETH (as the number of edges in the graph is $\tilde{O}(n)$).

    \paragraph{Mixed lower bound for \meetplusdiam}:
    Next we consider a slightly more involved construction for \meetplusdiam{}. For any integer $t$ define the following construction. Given an instance $A$ of the $OV$-problem construct $t-1$ copies of $A$: $A_1, \ldots, A_{t-1}$. For every $a\in A$ denote its copy in $A_j$ by $a_j$. We construct a matching between each consecutive pair of sets: for every $a\in A$ and $j=1, \ldots, t-2$ add an edge $(a_j, a_{j+1})$. Next, construct a set $U$ of $d$ vertices $u_1, \ldots, u_d$. For every $a\in A$ and every $i\in [d]$ such that $a[i]=1$ add the edges $(a_{t-1}, u_i)$ and $(u_i, a_1)$.

    In the NO-case we claim that $\meetplusdiam\leq t+2$. Indeed, for any $a_j\in A_j, b_\ell\in A_\ell$ assume w.l.o.g that $j\geq \ell$. Since $a,b$ are non-orthogonal there exists a coordinate $i$ for which $a[i]=b[i]=1$. Thus the following length $\leq t$ path connects $a_j$ to $b_\ell$: $a_j, a_{j+1}, \ldots, a_{t-1}, u_i, b_1, \ldots, b_\ell$. Therefore $\dplus(a_j, b_\ell)\leq t$. For any $u\in U, a_j\in A_j$ there exists some $b_1\in A_1$ such that $d(u, b_1) = 1$. Now, by the preceding argument, $d(a_j, b_1) \leq t$, therefore $\dplus(u,a_j)\leq t+1$. Finally, for any pair $u,u'\in U$ there exists $a_1,b_1\in A_1$ such that $d(u,a_1)=d(u',b_1)=1$. Since $d(a_1, b_1)\leq t$ we conclude that $d(u,b_1)\leq t+1$ and so $\dplus(u,u')\leq t+2$.

    On the other hand, in the YES-case we claim that $\meetplusdiam{\geq 2t}$. Let $a,b$ be a pair of orthogonal vectors and consider $\dplus(a_1, b_1) = d(a_1, x) + d(b_1,x)$ for some node $x$. If $x = u_i \in U$ and $d(a_1, x) = d(b_1, x) = t-1$ then both $a$ and $b$ must have $1$ at coordinate $i$, in contradiction to them being orthogonal. Thus, at least one of the distances from $a_1$ or $b_1$ to $x$ must be at least $2t-1$ (going around the loop twice) and so $\dplus(a_1, b_1) > 2t$. If $x = c_j\in A_j$ and $d(a_1, x)< t$ then $c = a$ as the path from $a_1$ to $c_j$ contains only matching edges, meaning $x = a_j$. Now, if $d(b_1, a_j) < 2t$ then the path from $b_1$ to $a_j$ consists of a series of matching edges, $b_1, \ldots, b_{t-1}$, followed by a vertex in $U$,  $u_i$, followed by a vertex $f_1\in A_1$ and another series of matching edges $f_1, \ldots, f_j = a_j$. Thus we must have $f=a$. Now, the edges $(b_{t-1}, u_i), (u_i, a_1)$ indicate that both $a$ and $b$ have a 1 at coordinate $i$, in contradiction to them being orthogonal. We conclude that  $d(a_1, x) < t$ implies that $d(b_1, x) \geq 2t$ and thus $\dplus(a_1, b_1) \geq 2t$ and $\meetplusdiam \geq 2t$.

    For any constants $\varepsilon, C>0$ we can pick large enough $t$ ($t > \frac{4+C-2\varepsilon}{\varepsilon}$) so that a $(2-\varepsilon, C)$-approximation differentiates between $\leq t+2$ and $\geq 2t$. Thus, a $(2-\varepsilon, C)$-approximation for the \meetplusdiam{} solves OV and requires $m^{2-o(1)}$ time under SETH. 

    \paragraph{Mixed lower bound for \meetmaxdiam}:
    Note that the above lower bound rules out not just $(2-\varepsilon)$-approximations to the \meetplusdiam{} but also ones with constant additive error. We now extend our lower bound for \meetmaxdiam{} to rule out such mixed approximations as well. Given an integer $t$, we replace the set $A$ with $t$ copies of it $A_1, \ldots, A_{t}$, adding a matching between $A_i$ and $A_{i+1}$ for $i=1, \ldots, t-1$. We then take a set $U$ of size $d$ representing the coordinates and add the edges from $a\in A_{t}$ to coordinates $u_i$ for which $a[i]=1$. We then add ${t-1}$ additional copies of $U$: $U_1, \ldots, U_{t-1}$. We add a matching between $U$ and $U_1$ and $U_j, U_{j+1}$ for $j=1, \ldots, t-2$. Lastly, we add all edges from $U_{t-1}$ to $U$. 
    
    The same argument shown above gives us that in the NO-case we have $\meetmaxdiam \leq t$ and in the YES-case we have $\meetmaxdiam \geq 2t$. Thus, by picking a large enough value of $t$, any $(2-\varepsilon, O(1))$-approximation to the \meetmaxdiam{} would solve OV and thus require $m^{2-o(1)}$ time under SETH.

\end{proof}

\begin{theorem}\label{thm:daglb}
    Under SETH, computing a $(3/2 - \varepsilon, O(1))$-approximation to the \meetplusdiam{} in a DAG requires $m^{2-o(1)}$ time.    
\end{theorem}

\begin{proof}
    Consider an OV instance $A$ and take $2t$ copies of $A$, $A_1, \ldots, A_{2t}$ with a matching from each $A_j$ to $A_{j+1}$. We then take a set $U_1$ of size $d$, $U_1=\{u_1, \ldots, u_d\}$ representing the coordinates and add the edges from $a\in A_{2t}$ to coordinates $u_i$ for which $a[i]=1$. We add an additional $t-1$ copies of $U_1$: $U_2, \ldots, U_t$ and a matching from each $U_j$ to $U_{j+1}$. Finally, we add a vertex $x$ and add an edge from every vertex in $U_t$ to $x$. We note that this construction is a DAG.

    In the NO-case of the OV-problem we have that any pair of vertices $a_1, b_1\in A_1$ have some coordinate $i$ for which $a[i]=b[i]=1$ and so by meeting at the copy of $u_i$ in $U_1$ the two vertices can meet within meet-distance $4t$. Thus, any pair of points  $u,w\in A_1\cup \ldots \cup A_{2t}$ has $\dplus(u,w)\leq 4t$. On the other hand, any point $w\in U_1\cup \ldots \cup U_t\cup \{x\}$ can meet any other point in the graph in distance $4t$ by meeting at $x$. We conclude that in the NO-case $\meetplusdiam\leq 4t$.

    In the YES-case we have a pair of points $a,b\in A$ such that $a\cdot b = 0$. Thus, by the same arguments we saw in previous lower bounds, $a_1, b_1$ cannot meet at $U_1\cup \ldots \cup U_t$ and so must meet at $x$ with $\dplus(a_1, b_1)=6t$. We conclude that $\meetplusdiam\geq 6t$.

    Thus, by picking a large enough value of $t$, any $(3/2-\varepsilon, O(1))$-approximation to the \meetplusdiam{} would solve OV and therefore, under SETH, requires $m^{2-o(1)}$ time.

\end{proof}

\subsubsection{Radius Lower Bounds}\label{sec:radiuslb}

\begin{theorem} \label{thm:lbradius}
    Under the Hitting Set Conjecture, computing a finite approximation to the \meetmaxradius{} or 
    the \meetplusradius{} in a DAG requires $m^{2-o(1)}$ time.
\end{theorem}

\begin{proof}
    Let $A,B$ be an instance of the hitting set problem. Take the sets $A, B$ and a set $U = \{u_1, \ldots, u_d\}$. For any $a\in A, i\in [d]$ such that $a[i]=1$ connect $a$ to $u_i$ and for any
    $b\in A, i\in [d]$ such that $b[i]=1$ connect $u_i$ to $b$. Add an additional vertex $x$ 
    with all edges from $A$ and an additional vertex $y$ with all edges from $A\cup U$.

    In the YES-case of the HS-problem, there exists $a\in A$ such that $\forall b\in B$ we have $a\cdot b \neq 0$. We claim that $\eccmax(a) = \eccplus(a) = 2$. Indeed, $a$ can meet any vertex in $A$ or $U$ at $y$ with \meetmaxdist{} 1 and \meetplusdist{} 2. It can reach $x,y$ directly in distance 1. For any $b\in B$ it can reach it in 2 steps via $u_i$ for some coordinate $i$ such that $a[i] = b[i] = 1$. Thus, $\dmax (a, b) = \dplus(a, b) = 2$ and we conclude that $\eccmax(a) = \eccplus(a) = 2$.

    In the NO-case, for any $a\in A$ there exists $b\in B$ such that $a\cdot b = 0$. Thus, $a$ cannot reach $b$ with a 2-hop path through $U$. However, this is the only type of path that could possibly allow $a$ to meet $b$ as $b$ is a sink and there are no edges going back into $A$ that would allow longer paths. Thus $\eccmax(a) = \eccplus(a) = \infty$. Furthermore, any vertex outside of $A\cup \{x\}$ cannot meet $x$ and thus has infinite eccentricity and the vertex $x$ cannot meet $y$ and has infinite eccentricity. We conclude that in the NO-case the graph has infinite meet-radius.

    Thus, any finite approximation to the \meetmaxradius{} or \meetplusradius{} would solve HS and under the HSC requires $m^{2-o(1)}$ time. 
    \end{proof}

\section{Extensions to $k$ Points}
We now study the analogous $k$-point \meetmaxdist{} and \meetplusdist{} in a directed graph; recall that our definition of $k$-point meet distance applies to any tuple of $k$ points (possibly with repetition).

We can trivially compute the exact $k$-point \meetmaxdist{} and
\meetplusdist{} diameters in
$O(mn+k n^{k+1})$ time by first computing APSP and then, for each
$k$-tuple, considering each of the $n$ possible meeting points. This gives us the exact values of $\kmeetmaxdiam{k}$ and $\kmeetplusdiam{k}$.

To improve upon this running time, first we observe that for \meetmaxdist{} we can use the $\sqrt{k}$-point diameter to approximate the $k$-point diameter.

\begin{claim}
    $\kmeetmaxdiam{\left\lceil\sqrt{k}\right\rceil}\leq \kmeetmaxdiam{k}\leq 2\cdot\kmeetmaxdiam{\left\lceil\sqrt{k}\right\rceil}$.
    \label{claim:k-2}
\end{claim}

\begin{proof}
    The first inequality follows immediately. The second inequality follows because any $k$ vertices can be partitioned into $\lceil\sqrt{k}\rceil$ components of size at most $\lceil \sqrt{k} \rceil$. Each vertex can first meet up with the other vertices in its component in time at most $\kmeetmaxdiam{\left\lceil\sqrt{k}\right\rceil}$, and a further routing connecting the $\left\lceil\sqrt{k}\right\rceil$ centers of the components will take time at most $\kmeetmaxdiam{\left\lceil\sqrt{k}\right\rceil}$.
\end{proof}

\begin{corollary}
    We can compute a 2-approximation to \kmeetmaxdiam{k} in time $O(mn + \sqrt{k} n^{\lceil\sqrt{k}\rceil + 1})$.
\end{corollary}

\begin{proof}
    The algorithm simply outputs $\kmeetmaxdiam{\left\lceil\sqrt{k}\right\rceil}$, which is computable exactly in time $O(mn + \sqrt{k} n^{\lceil\sqrt{k}\rceil+1})$. Correctness of approximation follows by Claim \ref{claim:k-2}.
\end{proof}

More generally, we have the following claim.
\begin{claim}
    Let $\ell$ be a positive integer.
    \[\kmeetmaxdiam{\left\lceil k^{1/\ell}\right\rceil}\leq \kmeetmaxdiam{k}\leq \ell \cdot \kmeetmaxdiam{\left\lceil k^{1/\ell}\right\rceil}.\]
    \label{claim:k-l}
\end{claim}

\begin{proof}
    The first inequality follows immediately. To prove the second inequality, we observe that, for any tuple of $k$ points, we can create a tree with each point as a leaf, maximum branching factor at most $\left\lceil k^{1/\ell} \right\rceil$ and height at most $\ell$. For each internal node, we choose the meeting point to minimize the maximum meeting time of its children, which is at most $\kmeetmaxdiam{\left\lceil k^{1/\ell} \right\rceil}$. By construction, every node in the tree (including the leaves) can reach the root in maximum distance at most $\ell \kmeetmaxdiam{\left\lceil k^{1/\ell}} \right\rceil$.
\end{proof}

\begin{corollary}
    We can compute an $\ell$-approximation to \kmeetmaxdiam{k} in time $O(mn + \left\lceil k^{1/\ell}\right\rceil n^{\left\lceil k^{1/\ell}\right\rceil + 1})$.
\end{corollary}

\begin{proof}
    It suffices to output $\kmeetmaxdiam{\left\lceil k^{1/\ell} \right\rceil}$, which can be computed exactly in time $O(mn + \left\lceil k^{1/\ell}\right\rceil n^{\left\lceil k^{1/\ell}\right\rceil + 1})$.
\end{proof}

We now consider the \meetplusdist{}. Here, the same argument gives us a seemingly weaker guarantee:
\begin{claim}
    $\kmeetplusdiam{k}\leq 2\left\lceil\sqrt{k}\right\rceil\cdot\kmeetplusdiam{\left\lceil\sqrt{k} \right\rceil}$.
\end{claim}

\begin{proof}
    We can again divide any given $k$ vertices into $\lceil \sqrt{k} \rceil$ groups of at most $\lceil \sqrt{k} \rceil$ vertices; each group can meet at a common point in total additive distance $\kmeetplusdiam{\left\lceil\sqrt{k} \right\rceil}$. Summing over all  $\lceil \sqrt{k} \rceil$ groups gives an additive distance of $\left\lceil\sqrt{k}\right\rceil\cdot\kmeetplusdiam{\left\lceil\sqrt{k} \right\rceil}$. These $\lceil \sqrt{k} \rceil$ centers can themselves meet in additive distance $\kmeetplusdiam{\left\lceil\sqrt{k} \right\rceil}$ at a root center. To route all $k$ points to the root center will thus take total additive distance $2\left\lceil\sqrt{k}\right\rceil\cdot\kmeetplusdiam{\left\lceil\sqrt{k} \right\rceil}$.
\end{proof}

However, by taking a $k$-tuple, with composition including $\lfloor k/(\lceil \sqrt{k} \rceil) \rfloor$ copies of the $\lceil \sqrt{k} \rceil$ points that achieve the $\kmeetplusdiam{\lceil \sqrt{k} \rceil}$, we have that 
\[
    \kmeetplusdiam{k} \geq \lfloor k/(\lceil \sqrt{k} \rceil) \rfloor \cdot \kmeetplusdiam{\lceil\sqrt{k}\rceil}.
\]
We observe as a natural consequence that for a square $k$, $\sqrt{k} \cdot \kmeetplusdiam{\sqrt{k}}$ is a $2$-approximation to $\kmeetplusdiam{k}$. We now extend this idea to larger $\ell$-approximations.

\begin{claim}
    Let $\ell$ be a positive integer, and let $k$ satisfy that $k^{1/\ell}$ is an integer. Then
    \[k^{1-1/\ell}\cdot \kmeetplusdiam{k^{1/\ell}}\leq \kmeetplusdiam{k}\leq \ell \cdot k^{1-1/\ell} \cdot \kmeetplusdiam{k^{1/\ell}}.\]
\end{claim}
\begin{proof}
    The first inequality follows since we can construct a $k$-tuple of points that is $k^{1-1/\ell}$ copies of the $k^{1/\ell}$ points with maximum meet distance of  $\kmeetplusdiam{k^{1/\ell}}$. This construction will have meet distance exactly $k^{1-1/\ell}\cdot \kmeetplusdiam{k^{1/\ell}}$.

    For the second inequality, we again construct a tree of depth at most $\ell$ and branching factor $k^{1/\ell}$, in which each internal node is the optimal meeting point of its children. We observe that to calculate the total additive distance for all leaves to reach the root, each level of the tree induces cost $k^{1 - 1/\ell} \kmeetplusdiam{k^{1/\ell}}$, and thus the total additive distance is $\ell k^{1 - 1/\ell} \cdot \kmeetplusdiam{k^{1/\ell}}$.
\end{proof}

\begin{corollary}
    We can compute an $\ell$-approximation to $\kmeetplusdiam{k}$, for k satisfying $k^{1/\ell}$ is integer, in time $O\left(mn + k^{1/\ell} n^{k^{1/\ell} + 1}\right)$.
\end{corollary}
\begin{proof}
    It suffices to output $k^{1-1/\ell}\cdot \kmeetplusdiam{k^{1/\ell}}$, which can be computed exactly in time $O\left(mn + k^{1/\ell} n^{k^{1/\ell} + 1}\right)$.
\end{proof}

Finally, we show how to obtain an $\ell$-approximation for an arbitrary integer $k$.
\begin{corollary}
    For every positive integer $\ell$, we can compute a value $\widehat D$ satisfying
    \[
        \widehat D\leq \kmeetplusdiam{k}\leq \ell\widehat D
    \]
    in time
    \[
        O\left(nm+\ell k^{1/\ell} n^{\lceil k^{1/\ell} \rceil +1}\right).
    \]
    \label{cor:ell}
\end{corollary}

\begin{proof}
   For any integer $\ell\geq 1$, let
    \[
        s=\left\lceil k^{1/\ell}\right\rceil.
    \]
    Since $s^\ell\geq k$, there is a rooted tree $T$ of height $\ell$
    with exactly $k$ leaves, all at depth $\ell$, and with at most $s$
    children per node.

    For each node $v$, let $w(v)$ be the number of leaves below $v$.
    If the children of $v$ are $v_1,\ldots,v_t$, define
    \[
        \Delta_v
        :=
        \max_{x_1,\ldots,x_t\in V}
        \min_{z\in V}
        \sum_{i=1}^t w(v_i)d(x_i,z).
    \]
    For $j=1,\ldots,\ell$, let
    \[
        A_j
        :=
        \sum_{\substack{v\in T\\\operatorname{depth}(v)=j-1}}
        \Delta_v,
    \]
    and output
    \[
        \widehat D:=\max_{1\leq j\leq\ell}A_j.
    \]

    First, we show $A_j\leq\kmeetplusdiam{k}$ for every $j$. Indeed, for each
    node $v$ at depth $j-1$, take points attaining $\Delta_v$ and
    replace the point corresponding to a child $v_i$ by $w(v_i)$
    copies. Concatenating these tuples over all nodes at depth $j-1$
    gives a $k$-tuple whose cost to every common center is at least
    \[
        \sum_{\operatorname{depth}(v)=j-1}\Delta_v=A_j.
    \]

    Conversely, place an arbitrary $k$-tuple on the leaves of $T$ and
    merge it upward. At an internal node $v$, after the points below
    each child $v_i$ have met at some $x_i$, move them to their best
    common center. The additional cost is
    \[
        \min_{z\in V}\sum_{i=1}^t w(v_i)d(x_i,z)
        \leq \Delta_v.
    \]
    After all $\ell$ levels, every point has reached the root.
    Therefore,
    \[
        \kmeetplusdiam{k}
        \leq \sum_{v\text{ internal}}\Delta_v
        =\sum_{j=1}^{\ell}A_j
        \leq \ell\widehat D.
    \]
    Together with $\widehat D\leq\kmeetplusdiam{k}$, this proves the
    approximation guarantee.

    Finally, take \(T\) to be the subtree spanned by the first \(k\) leaves of the complete \(s\)-ary tree of height \(\ell\). At each depth, every internal node except possibly one is complete, so there are at most two child-weight vectors at that depth: one for the complete nodes and one for the unique boundary node. Thus, rather than evaluating \(\Delta_v\) separately for every node, we evaluate only \(O(\ell)\) distinct weighted instances and multiply each value by the number of nodes of its type. Each instance involves at most \(s\) centers and one meeting point, and hence can be computed by exhaustive enumeration in \(O(s n^{s+1})\) time. The level sums \(A_1,\ldots,A_\ell\), and therefore \(\widehat D=\max_j A_j\), can then be assembled in \(O(\ell)\) time. Including the \(O(nm)\) time needed to compute all pairwise distances, the total running time is
\[
    O\!\left(nm+\ell s n^{s+1}\right).
\]
\end{proof}

\subsection{Small Hitting Sets for Unit Meet-Diameter}

\begin{theorem}
    Suppose that a graph $G$ has $\meetmaxdiam$ $1$ (for every pair of nodes $u, v$, there is a common out-neighbor $x$). Then there is a set $S$ of size $O(\sqrt{n} \log n)$ such that every node has an out-neighbor in $S$.
\end{theorem}

\begin{proof}
    Let $L$ be our degree threshold. We let $S$ initially be a hitting set of size $O((n/L) \log n)$, which covers the out-neighborhood of all vertices with out-degree at least $L$. Let $R$ be the set of nodes with no out-neighbor in $S$. For all $r \in R$, the out-degree of $r$ is at most $L$.

    We then perform the following iterative process. Take an arbitrary $r \in R$, and let $x$ be the out-neighbor of $r$ with the most in-neighbors in $R$. Since $r$ has meet distance $1$ with everything in $R$, the in-degree of $x$ in $R$ must be at least $|R|/\deg_{\mathrm{out}}(r) \geq |R|/L$. We add $x$ to $S$ and remove all of its in-neighbors from $R$. The new size of $R$ is then at most $|R|(1 - 1/L)$.

    After $L \ln|R| \leq L \ln n$ iterations of the above procedure, the size of $R$ is $|R|(1 - 1/L)^{L \ln |R|} < 1$; thus $R$ is the empty set. Therefore we obtain a hitting set $S$ of size $(L + (n/L)) \log n$. Setting the degree threshold to be $L = \sqrt{n}$ gives the desired result.
\end{proof}

\section*{AI Disclosure}
We used ChatGPT (OpenAI) to check our writing and correctness. The tool materially affected \cref{sec:radiuslb} by improving our meet-radius lower bound from a constant to any finite approximation. It also affected the proof of \corref{ell} by extending the $\ell$-approximation to apply to all $k$ (including $k$ with non-integer $k^{1/\ell}$) without any (even low-order) loss in the approximation. The authors verified the correctness and originality of all content including references.

\section*{Acknowledgments}
This research was supported by NSF Grant CCF-2330048, BSF Grant 2024233, and a Simons
Investigator Award. Y. Kirkpatrick was supported in part by an NSF Graduate Research Fellowship
under Grant No 2141064. J. Kuszmaul was supported in part by an MIT Akamai Presidential
Fellowship.

\bibliographystyle{plainurl}
\bibliography{meet}

\end{document}